\documentclass[12pt, a4paper]{article}
\usepackage{amsmath}
\usepackage{amsthm}
\usepackage{import}
\usepackage{bbm}
\usepackage{tikz}
\usepackage{geometry}
\usepackage{enumerate}
\usepackage{comment}
\usepackage[onehalfspacing]{setspace}  
\usepackage[title]{appendix}
\usepackage[utopia]{mathdesign}

\usepackage{fullpage}

\usepackage{sectsty}
\sectionfont{\centering \Large \scshape }
\subsectionfont{\centering \large  \scshape}
\usepackage[american]{babel}
\usepackage{csquotes}
\usepackage[style=apa, uniquelist=false, backend=biber, doi=false,isbn=false,url=false]{biblatex}
\DeclareLanguageMapping{american}{american-apa}
\defbibenvironment{bibliography}
{\enumerate{}
{\setlength{\leftmargin}{\bibhang}%
\setlength{\itemindent}{-\leftmargin}%
\setlength{\itemsep}{\bibitemsep}%
\setlength{\parsep}{\bibparsep}}}
{\endenumerate}
{\item}
\usepackage[colorlinks,citecolor=blue]{hyperref}

\usetikzlibrary{shapes.geometric}
\usetikzlibrary{positioning}


\newcounter{parentnumber}
\theoremstyle{plain}
\newtheorem {theorem}{Theorem}

\newtheorem{assumption}{Assumption}

\newtheorem {definition}{Definition}
\newtheorem {observation}{Observation}
\newtheorem {remark}{Remark}
\newtheorem {lemma}{Lemma}
\theoremstyle{definition}
\newtheorem {example}{Example}

\newtheorem{proposition}[theorem]{Proposition}
\usepackage{color}

\newcommand{\R}{\mathbb{R}}

\title{Pricing and Electric Vehicle Charging Equilibria}

\author{
    Trivikram Dokka \footnote{Advanced Analytics Group, Air Products Plc, United Kingdom. Email:
    \href{mailto:trivikram.dokka@yahoo.co.uk}{trivikram.dokka@yahoo.co.uk}}
    \ \and
     Jorge Bruno\footnote{Department of Digital Technologies, Faculty of Business and Digital Technologies,
University of Winchester. Email:\href{mailto:Jorge.Bruno@winchester.ac.uk}{Jorge.Bruno@winchester.ac.uk}}\ \and
    Sonali SenGupta\footnote{ Economics Section, Queens Management School,
    Queens University Belfast, United Kingdom. Email:
    \href{mailto:s.sengupta@qub.ac.uk}{s.sengupta@qub.ac.uk}} \ \and 
    Chowdhury Mohammad Sakib Anwar\footnote{\textit{Author for correspondences.} Department of Economics, University of Winchester, United Kingdom. Email:
    \href{mailto:Sakib.Anwar@winchester.ac.uk}{Sakib.Anwar@winchester.ac.uk}.}
}
\date{\normalsize \today}

\begin{document}

\maketitle
\begin{abstract}
We study equilibria in an Electric Vehicle (EV) charging game, a cost minimization game inherent to decentralized charging control strategy for EV power demand management. In our model, each user optimizes its total cost which is sum of direct power cost and the indirect dissatisfaction cost. We show that taking player specific price independent dissatisfaction cost in to account, contrary to popular belief, herding only happens at lower EV uptake. Moreover, this is true for both linear and logistic dissatisfaction functions. We study the question of existence of price profiles to induce a desired equilibrium. We define two types of equilibria, distributed and non-distributed equilibria, and show that under logistic dissatisfaction, only non-distributed equilibria are possible by feasibly setting prices. In linear case, both type of equilibria are possible but price discrimination is necessary to induce distributed equilibria. Finally, we show that in the case of symmetric EV users, mediation cannot improve upon Nash equilibria.

\medskip
\begin{flushleft}\textbf{Keywords} : Electric Vehicles, Pricing, Nash equilibrium, Coarse correlated equilibrium, Mediation, Herding, Dissatisfaction cost

\end{flushleft}\par
\begin{flushleft}\textbf{JEL Codes}: C61, C72, D4, D11, D82 \ \end{flushleft}\par
\end{abstract}

\newpage
\section{Motivation and research questions}\label{intro_sec}

Electric Vehicles (EVs) are widely seen as part of a solution to economically and environmentally sustainable transportation future. With more countries looking to de-carbonize their economies at an increased pace, more incentives to EV uptake are being proposed and implemented. However, mass scale EV uptake comes with its own challenges, primary of them being impact on existing electricity infrastructure. Several researchers investigate economic and environmental implications of residential charging of electric vehicles \parencite{clement2009impact,muratori2018impact}.  It is widely believed that an efficient demand response is essential to avoid costly infrastructure upgrades and/or blackouts. This involves alignment of EV charging demand with supply. Such an alignment not only avoids costly and unnecessary capacity addition but also results in shift to renewable sources. Initial ideas to achieve this alignment, naturally, were based on incentivizing people via prices (to charge at non-peak times)\parencite{palensky2011demand,electricnation}. It is argued that the same price signals received by all EV users will result in herding behavior where, all users shift their charging to low cost periods to avoid peak periods, resulting in new peaks \parencite{valogianni2020sustainable,electricnation,electricnation2019finalreport}. Herding behavior formation argued in these studies relies directly on the assumption that EV users are cost minimizers. Ironically, the observations of herding are made without any reference to the level of EV uptake. In a typical herding scenario (with high uptake), a EV user receives less power (kWh) than is expected when charged at full speed, due to congestion or as part of a strategy to move users to different time of day, hence causing dissatisfaction due to less battery charge received. Without taking this dissatisfaction in to account the conclusion that charging behaviors result in herding may not be consistent with expected behavior over time. The motivation for capturing dissatisfaction explicitly is justified because price as an instrument to control charging behavior is only possible when users are (or are not) willing to pay to avoid dissatisfaction. In fact, more recently user dissatisfaction is explicitly modeled within an algorithmic charging decision-making set-up \parencite{lin2021minimizing}. Similarly, \textcite{wu2022smart} uses the term inconvenience cost in the same sense and illustrate optimal mechanisms for EV charging at public stations.  Consistent with this, our first research question is:
\paragraph{Question 1} \quad \textit{When users experience indirect costs associated with dissatisfaction will herding still happen?}

Grid managers and DNOs would want to use the flexibility of EV charging (believed to be flexible load compared to other loads such as household power demand) to achieve a desired consumption profile which better aligns with grid management objectives (see \textcite{valogianni2020sustainable} and references therein). Recent research suggests designing a dynamic and adaptive pricing schemes to achieve a desired charging profile  \parencite{jacobsen2022consumers}. Our goal differs from these works in that we seek to find if price profiles exist which will lead to a desired behavior profile in equilibrium, and if so under what conditions, taking into account congestion \textit{aka} dissatisfaction. To the best of our knowledge we are not aware of any work that takes into account dissatisfaction cost or analyzes equilibrium outcomes. With this motivation, we address our second research question:
\paragraph{Question 2} \quad \textit{Does there exist price profiles that will induce a desired charging profile under the given player-specific price independent dissatisfaction costs?}

To answer our questions we take a game-theoretic approach to model EV user's selfish behavior and use a stylized model that captures the key aspects of EV charging behavior as an EV charging game. Centralized versus decentralized control of charging has received much attention in EV related literature. Decentralized setting can be seen through the game theoretic lens, an approach only taken by relatively few researchers compared to much more abundant empirical studies. While earlier studies, such as \textcite{tushar2012economics}, considered a stackelberg approach, closer to our setting, simultaneous form games were studied in \textcite{chakraborty2014demand,chakraborty2017distributed}. However, no studies considered dissatisfaction cost within game-theoretic setting. Our work also complements the alternative stream of literature that takes a mechanism design approach to pricing problem, such as \textcite{wu2022smart,nejad2017online}.

Price-based or otherwise, the idea of (decentralized) controlled charging relies on the assumption that EV users will find themselves better off when an agency such as EV aggregator acts as recommender system; under the belief that such an entity has greater (technological and informational) ability to make better charging recommendations compared to EV users deciding on their own.\footnote{A number of researchers proposed optimization algorithms under decentralized scenario, see \textcite{shen2019optimization}.} The role of mediating agency is certainly not unique to this situation, and many economic situations, whether it be resource sharing or contributing towards a public good, also have this characteristic. Theoretically, such an entity will recommend (or implement), with user's consent, a charging regime which may or may not satisfy user's complete demand but may result in lower cost. But, what if users do not find following agency recommendations better than their own decisions? Therefore, a confirmation of existence of such an entity is necessary via game theoretic analysis. This leads to our final question.
\paragraph{Question 3} \quad \textit{Will co-ordination (or co-ordinated mediation) help? In other words, if an agency (aggregator/charging manager) acts as a recommender system of how to charge, will EV users commit to such an agency, and if they do, does it lead to a different equilibrium than if they do not?} 

It is well known that in non-cooperative settings mediated communication is an efficient way to achieve incentive-compatible outcomes via correlation devices \textit{aka} correlated equilibrium \parencite{aumann1987correlated} and Coarse Correlated Equilibrium (CCE)\parencite{moulin1978strategically}. We adopt CCE to answer if mediated communication leads to different outcomes as against when EV users behave on their own. CCE, in recent years has received considerable attention owing to the finding that no-regret play leads to coarse-correlated equilibria \parencite{roughgarden2015intrinsic}. Correlated equilibria have also been associated with evolutionary learning \parencite{arifovic2019learning}.


From the structure of games point of view, the games that we study in our work could be seen to be connected to congestion and budget games, hence a comment on connection to this literature is in order. The extant literature on congestion games spanning areas of economics, computer science and operations research fields, mainly focus on existence and efficiency of equilibria, that too, predominantly (pure) Nash equilibria. For example, these results commonly establish bounds on price of anarchy and stability.  On the other hand, our questions are not related to efficiency of equilibria, instead, our questions are motivated by practical observations from EV field trials. In the context of our games, a desirable outcome may not even be the efficient one as is usually defined. It is conceivable that games in our work could be modeled via congestion games frameworks, furthermore, efficiency questions may also be relevant (as discussed in \textcite{chakraborty2014demand}).  However, this is not the main focus of our work and we leave it for future study. 

The rest of the paper is organized as follows. In Section \ref{model} we present a discrete EV charging game along with the main assumptions. In Section \ref{results}, we outline our three main results that answer the questions stated in Section \ref{intro_sec}. In Sections \ref{herding_sec}, \ref{pricing_sec}, and \ref{cce_sec} we present the details including statements and proofs of the results underlying research questions 1, 2 and 3 respectively. 

\section{EV charging game: model and assumptions}\label{model}

Consider the following game we call \textit{EV charging game (EVCG)}.  There are $n$ players. A typical Player $i$ has a demand $r_id_i$ which can be fulfilled by choosing to charge in any of $d_i \leq T$ time periods. That is, strategy of a player is a $T$-dimensional binary vector.

Given a strategy profile $(s_i, s_{-i})$, the dis-utility/cost of Player $i$ is 

\begin{equation}
    c(s_i,s_{-i}) = \sum_{t} b_t g_t(s_i, s_{-i}) + \sum_{t}  f_t(s_i, s_{-i}),
\end{equation}

where 

\[g_t(s_i, s_{-i})=
\begin{cases}
  \frac{P_ts^t_i}{\sum_{k=1}^n r_ks^t_k}, & \quad \mbox{if $ \sum_{k=1}^n r_ks^t_k>P_t$}\\
     r_is^t_i, & \quad \mbox{otherwise}
\end{cases}
\]

\[ f_t(s_i, s_{-i})=
\begin{cases}
  h\left (s^t_i\left(\frac{P_t}{\sum_{k=1}^n r_ks^t_k}\right)\right), &\quad \mbox{if $\sum_{k=1}^n r_ks^t_k>P_t$}\\
     0, & \quad \mbox{otherwise}
\end{cases}
\]
with $r$ and $P$ being parameters of the game which are explained as follows: $r_i$ is the power rating (in $kW$) of Player $i$ which informs power transfer rate; $P_t$ is the total available power (in $kWh$) in time period $t$; $s_i^t$ indicates whether Player $i$ charges at time period $t$ or not; $\sum_{k=1}^n s^t_k$ is the total number of users who decided to charge in time period $t$; and $h$ represents a dissatisfaction function. Furthermore, each time period is classified according to $k$ price slabs. In the most general case, $k=T$. For this reason, typically, charging choices are modeled as discrete to allow for modeling a complete price discrimination between users, where users may pay time of day tariffs. However, consumers rarely choose schemes which employ real-time pricing or several price slabs during the day, this is usually explained as fear about price volatility. Most common practice is to differentiate between peak and non-peak times (see \textcite{electricnation,jacobsen2022consumers}). In our model, $b_t$ is the price per unit in time $t$, we will assume a two part price plan as formalized in the assumption below.
\begin{assumption}
Two price slabs: peak and non-peak; we superscript peak and non-peak times with $D$ and $N$ respectively, eg., $b^N_i$ is the non-peak price for user $i$.
\end{assumption}


\begin{remark}
Note that the scenario when $g_t(s_i, s_{-i})< r_i$ can be interpreted as congestion or a deliberate delayed charging strategy as in \textcite{wu2022smart} to induce a desirable charging behavior equilibrium.
\end{remark}
\begin{assumption}
$h(\cdot)$ is a continuous and monotone function. 
\end{assumption}

In our analysis we consider two functions: a linear and a logistic one. Linear dissatisfaction is also considered in the recent literature, the reason being linear dissatisfaction rates are more appealing because of simplicity and associated tractability of resulting analysis. In practice, however, it is likely the dissatisfaction is different across the support. For example, increase in dissatisfaction is probably higher when a user gets 10 kWh instead of 11kWh than when they get 20 kWh instead of 21 kWh. In both these situation the dissatisfaction caused in linear case would be same but in practice this is usually not the case. For this reason, we consider logistic dissatisfaction which captures change in dissatisfaction being different at lower and upper quantiles of power distribution. 

\begin{assumption}
$r_i = 1$, $P_t\geq 1$ and $\sum_t s_i^t = d_i$.
\end{assumption}

The assumption that all charging rates are equal is for sake of simplicity, and since most home chargers and rapid public chargers provide very similar rates of charge. All the results can be adapted when this assumption is relaxed. Furthermore, $\sum_t s_i^t > d_i$ is not considered because charging more time than necessary is captured via dissatisfaction.


To conclude this section, we introduce the notion of distributed equilibria in the following definition.
\begin{definition}
A strategy profile is of \emph{distributed} type when there is at least one user choosing to charge at both peak and non-peak periods and is of \emph{non-distributed} type, otherwise. Similarly, a \emph{distributed equilibrium} is one arising from a strategy profile of distributed type and it becomes a \emph{non-distributed equilibrium}, otherwise.
\end{definition}

\section{Results}\label{results}
\begin{enumerate}

    \item \textit{Question 1: Herding?} \quad We find that when users selfishly behave to minimize their overall cost, including dissatisfaction cost, herding is unlikely to happen in a more congested scenario. That is, considering a fixed power capacity, when EV uptake goes beyond a certain level users are more likely to converge to a distributed charging scenario. This is true with both linear and logistic dissatisfaction functions, and also true regardless of price discrimination. More specifically, in Theorem \ref{thm1} we show that there is EV uptake cut-off above which (some) users may find it profitable to deviate from charging only in peak time and distribute between peak and non-peak times.
    
    \item \textit{Question 2: Desired outcome inducing prices?} \quad Our results suggest charging behavior strongly depends on the dissatisfaction function. We illustrate that it is not possible to induce any arbitrary type of equilibrium by changing prices, even with complete price discrimination between users. That is, in spite of pricing differently for different players some behaviors may not be achievable. Hence, our work illustrates the limits of price based control. In Theorem \ref{pricing_thm}, we show that the in linear case both distributed and non-distributed equilibria are possible and for distributed profile to be equilibrium there is a unique price profile. However, in logistic case the only equilibria irrespective of prices are of non-distributed type. 
    
    \item \textit{Question 3: Aggregator or not?} \quad In many economic situations mediated communication has been shown to lead to better outcomes. However, the opposite is also true to some situations. We show that this is the case in one such situation when all users have same dissatisfaction and hence are charged the same prices - \textit{the symmetric case}. In Theorem \ref{main_cce_result}, we show that users are no better off than behaving selfishly on their own (i.e., Nash equilibria) compared to when a mediating agency recommends them how to charge (i.e., Coarse Correlated equilibria). We show that they coincide in this case. However, this result is only true for symmetric case. 
\end{enumerate}

Following the proofs of our results we discuss the implications of these results in the form of remarks.

\section{Herding}\label{herding_sec}

Let $T^N$ non-peak periods available for charging. Given a strategy $s_i \in \{0,1\}^T$, we allocate the first $ T^N$ entries of $s_i$ to the non-peak charging periods.  Cost function for player $i$ is

\begin{equation}\label{eqn:utility}
 c(s_i,s_{-i}) = \sum_{t=1}^{T^N} \left((s_i)_t x_t b_i^N  +  f_i^N(x_t)\right) + \sum_{t=T^N+1}^{T} \left((s_i)_t x_t b_i^D  +  f_i^D(x_t)\right)
\end{equation}

\noindent
where $f_i^D(x_t)$ and $f_i^N(x_t)$ represent Player $i$'s peak-time and non-peak-time dissatisfaction functions during period $t$, resp., and $x_t=\frac{P_t}{\sum_{j=1}^n (s_j)_t}$. The following result is then straightforward to verify.

\begin{observation}
From the definition of a Nash equilibrium if there is profitable, in this case, lower cost deviation, that is, a strategy profile $(s_1,\ldots, s_n)$ is not in equilibrium if, and only if, for some $i\leq n$ there exists distinct $t',t^*\leq T$ with $(s_i)_{t'}=1$, $(s_i)_{t^*}=0$ and 
\begin{equation}\label{eqn:main}
    x_{t'}b_i^{Q'} + f_i^{Q'}(x_{t'}) > x_{t^*}b_i^{Q^*} + f_i^{Q^*}(x_{t^*})
\end{equation}
where $Q',Q^* \in \{N,D\}$ 
\end{observation}


Recall that our aim is to establish if herding can be an equilibrium. To this end, for non-peak herding to be possible we require that $d_i\leq T^N$. The case when there is enough non-peak capacity (i.e., $\sum_i d_i \leq \sum_{t<T^N} P_t$) trivially leads to herding as in such a case users do not experience dissatisfaction.  We let $d_i = d_j= d $ for all $i,j \leq n$ and $P_t=P$ for all $t\leq T^N$ leading to a (demand, capacity)-homogeneity. Considering a homogeneous case allow us to focus our analysis on herding versus non-herding by removing all possibilities of several herding equilibria. In fact, one could arrive at the same findings as in Theorem 2 by instead considering period-specific constrained capacities ($P_t$) and heterogeneous demands ($d_i\not = d_j$) and then constructing an equivalent instance of homogeneity where all users incur same cost by adjusting the parameters of their respective dissatisfaction functions. Hence, no generality is lost from assuming (demand, capacity)-homogeneity.

Consider the strategy profile $(s_i,s_{-i}) = S_N$ where all $n$ players play the herding strategy with $T^N = d$ or $\sum_{k=1}^{T^N} (s_i)_k = d$ for all $i\leq n$.  It follows that $x_t=\frac{P}{\sum_{j=1}^n (s_j)_t}$ is the same for all $t\leq d$. Set this value as $x$. 

\begin{remark}
Note that $x$ is an indicator of system congestion, that is, it decreases with increase in the number of users $n$ at a fixed system capacity of $P_t$.
\end{remark}

\begin{theorem}\label{thm1}
There exists $\hat{x}\in (0,1)$ (threshold congestion level) such that for all $x<\hat{x}$ herding cannot be a Nash equilibrium when 
\begin{enumerate}
    \item \textit{Linear dissatisfaction} \quad $f_j^N(x) = \alpha_j - \beta_j x$, and when
    \item \textit{Logistic dissatisfaction} \quad $f_j^N(x) = \alpha_j\left(\frac{1}{1+e^{\beta^N_j\left(2x-1\right)}}\right)$.
\end{enumerate}
Furthermore, in the case of linear dissatisfaction the threshold congestion level is given by $x \leq \frac{\alpha_j - b_j^D}{\beta_j - b_j^N}$.
\end{theorem}

\begin{proof}

Suppose Player $j$ considers deviating from the standard herding strategy to a different strategy $s'_j$ and put $D_j = \sum_{t=d+1}^T s_j' > 0$. In which case, applying Equation~\ref{eqn:main}, we have that Player $j$ would deviate provided that 

\begin{equation}\label{eqn:ineq}
f_j^N(x) > b_j^D-b_j^Nx
\end{equation}
where we let $f_j^D(x_t) = 0$ for all $t>d$ as there is enough power available during any peak-time period since that Player $j$ would be the only player charging during a peak-time period. It is interesting to notice that the above inequality makes no reference to $D_j$. Thus, the utility difference by switching from herding is the same for any deviation.  

\textit{Linear case} \quad  Consider Player $j$'s linear non-peak-time dissatisfaction function $f_j^N: \R \to [0,1]$ for which $f_j^N(x) = \alpha_j - \beta_j x$ when $x<1$ and $0$, otherwise. From Equation \ref{eqn:ineq} we obtain that Player $j$ would deviate from herding provided that $\alpha_j - \beta_j x \geq b_j^D-b_j^Nx$. Since  $b_j^D > b_j^N$ then $f_j^N(1) = 0 < b_j^D-b_j^N$. It follows that for the cases where $\alpha_j < b_j^D$ there is no deviation from herding since both $f_j^N(0) <b_j^D$ and $f_j^N(1) <b_j^D-b_j^N$.\\


Alternatively, if $b_j^D < \alpha_j$ then $\beta_j>b_j^N$ since $\alpha_j - \beta_j = f_j^N(1) = 0 < b_j^D-b_j^N$. In turn, deviation from herding occurs for all values of 
\[
x \leq \frac{\alpha_j - b_j^D}{\beta_j - b_j^N}.
\]
This threshold point is illustrated in Figure \ref{fig:1}.

\begin{figure}[h]
    \centering
    \includegraphics[scale=0.95]{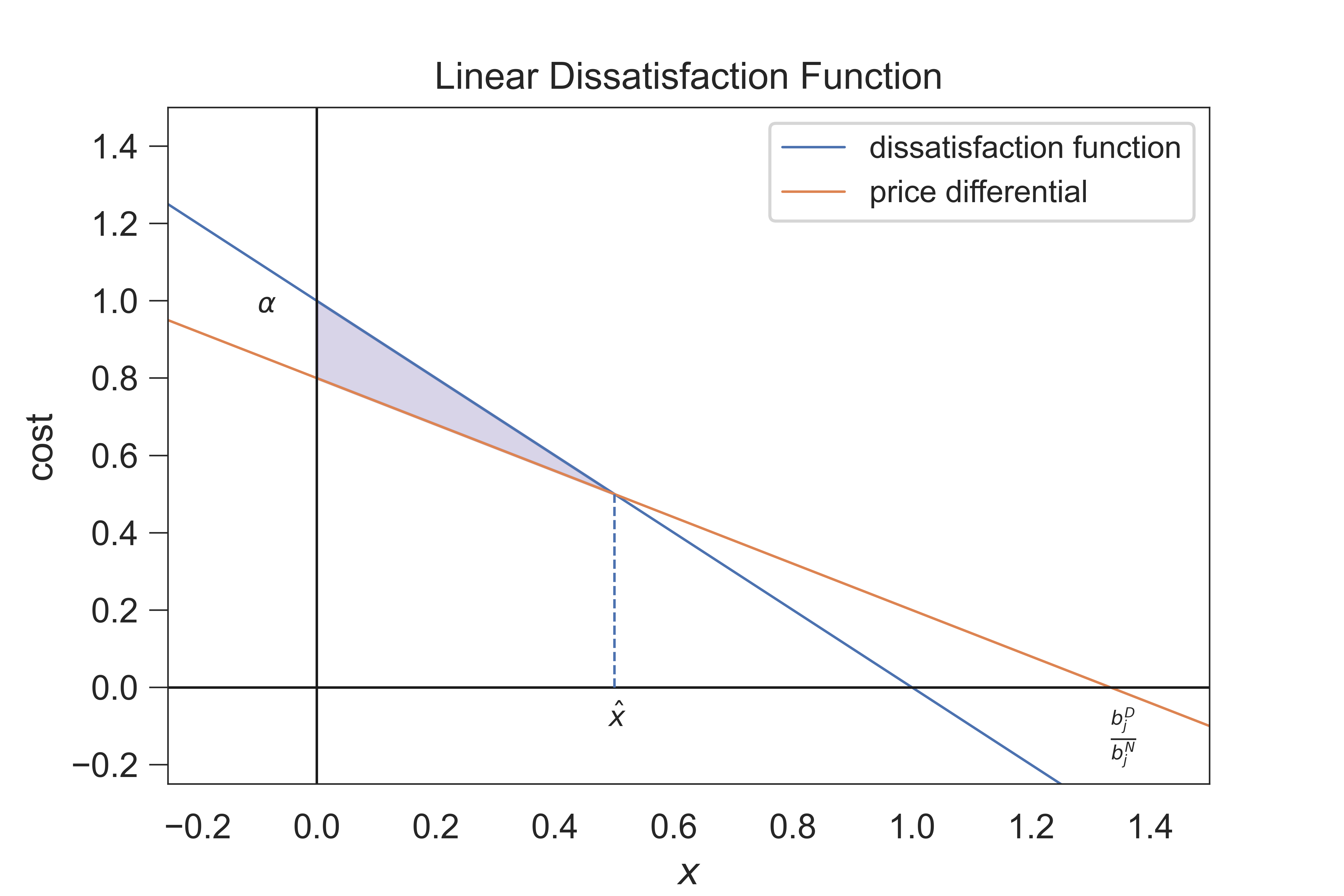}
    \caption{Switching to Non-Herding}
    \label{fig:1}
\end{figure}

\textit{Logistic dissatisfaction} \quad A more realistic dissatisfaction function can be expressed in terms of the logistic function as
\[
f_j^N(x) = \alpha^N_j\left(\frac{1}{1+e^{\beta^N_j\left(2x-1\right)}}\right).
\]
Observe that $f_j^N(0) = \alpha_j$ and $f_j^N(1) \approx 0$. As above, we have that $f_j^N(1) \approx 0 < b_j^D-b_j^N$ and since $b_j^D < \alpha_j$ then there exists a unique $\hat{x} \in (0,1)$ with 
\[
\alpha^N_j\left(\frac{1}{1+e^{\beta^N_j\left(2\hat{x}-1\right)}}\right) = b_j^D-b_j^N\hat{x}.
\]
By continuity of both $f_j^N(x)$ and $b_j^D-b_j^Nx$ it follows that for all $x \leq \hat{x}$ deviation from herding is preferred. Indeed, consider the unique inflection point $\overline{x} \in (0,1)$ of $f_j^N(x)$: $f_j^N(x)$ is concave down for all $x \in (0,\overline{x})$ and concave up for all $x \in (\overline{x}, 1)$. Considering the scenarios with $\overline{x} \leq \hat{x}$ and $\overline{x} > \hat{x}$ separately, it follows that $\hat{x}$ is the only intercept of $f_j^N(x)$ and $b_j^D-b_j^N x$.

This is illustrated in Figure \ref{fig:2}, where the shaded region illustrates the case where users deviate from herding behavior. Note that $b_j^D>\alpha_j$ is not realistic as in such a case trivially no user will charge in peak time.
\begin{figure}[h]
    \centering
    \includegraphics[scale=0.95]{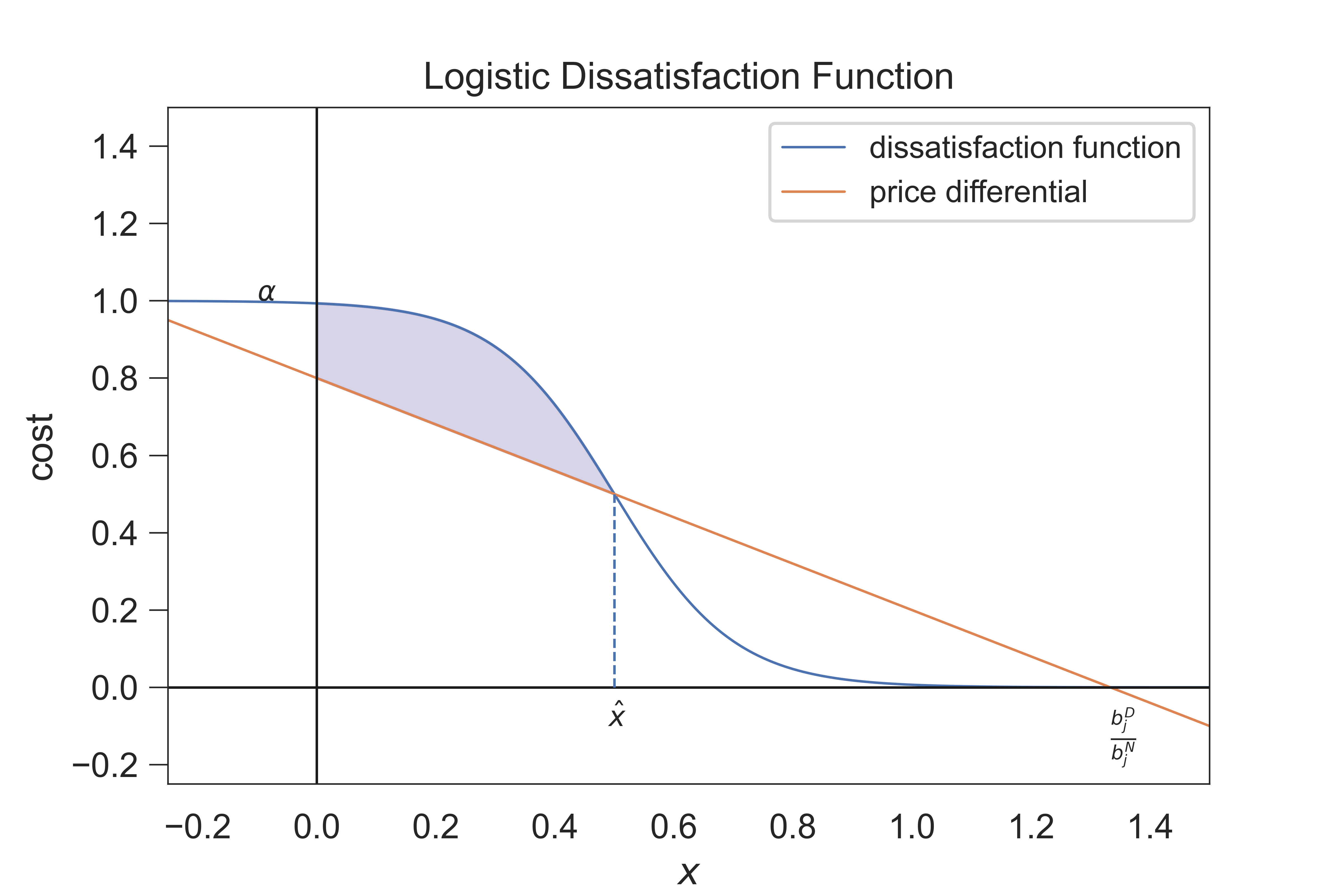}
    \caption{Switching to non-herding}
    \label{fig:2}
\end{figure}

\end{proof}

\begin{remark}[Implication of Theorem \ref{thm1}]
Theorem \ref{thm1} implies that at a fixed dissatisfaction rate and low EV uptake herding is likely to happen. Once the EV uptake increases (i.e., system congestion increases) herding is less likely to occur. This explains the early field trial observations where users moved to low cost periods in response to static price changes - reason being low system congestion. Naturally, low congestion is easier to manage with simpler price changes. In contrast, in a high EV uptake scenario a purely herding behaviour might not be achievable, even if desirable. There is also the possibility of reverse causation in a highly congested system with high dissatisfaction rates, $\beta > b$, and vice versa. 

Even without congestion, when delayed charging is employed as a strategy, our result shows the delay levels required to change charging behaviors.
\end{remark}

\begin{remark}
The above analysis leads to the question of whether or not the effect seen in Figures \ref{fig:1} and \ref{fig:2} is true in general for any monotone dissatisfaction function satisfying certain boundary conditions. The following example shows that this is not the case. Consider $f(x)=4-x^2$ with $b^D=4$ and $b^N=2$. In this case, herding is always deviated. Similarly, it is easy to construct a piece-wise monotone decreasing function where herding is never deviated. 
\end{remark}


Next we focus on illustrating two examples of non-herding Nash equilibria, one for each type of the dissatisfaction function.

\begin{example}[Linear Dissatisfaction] 


Consider the game with $n=3$, $P_t=2$,$d_i=2$, $T^N=2$, $T^D=2$. Assume Player 1 has $b^N_1=0.2$ and $b^D_1=0.3$, and Player 2 and 3 have $b^N_2=b^N_3=0.3$ and $b^D_2=b^D_3=0.4$. The dissatisfaction function for all players is $f(x)=1-x$. It is easy to show that in this scenario there is an equilibrium where Player 1 only charges in peak-time, and players 2 and 3 only charges non-peak time. The cost for Player 1 when charging in peak time only is 0.6, and the the costs for players 2 and 3 for charging in off-peak time  only is 0.6. Player 1 will not deviate to charging off-peak time since his cost increases to 0.7666. Similarly, players 2 and 3 will not deviate to peak-time since their costs increases to 1.
\end{example}

\begin{example}[Logistic Dissatisfaction] 

Consider the same scenario in the previous example except that the dissatisfaction is now logistic, $ f(x) = \frac{1.5}{1+e^{5(2x-1)} }$. As before we have an equilibrium where Player 1 charges in peak time, and players 2 and 3 charges in off-peak time. In fact in the following section we show that equilibria in logistic case are only of non-distributed type. 
\end{example}

\section{Pricing}\label{pricing_sec}


 We now turn our attention to our second research question. Before stating our result, we rewrite the equilibrium condition in the following convenient form. With slight abuse of notation, we will refer to $x^D_i$ for both amount of power of received and also proportion of power received against required (expected) power under the strategy profile $X$. Given strategy profiles $X$ and $Y$, Player $i$ will choose profile $X$ over $Y$ provided that
\begin{equation}\label{eq_cond_pricing}
b^D\Delta_i^D(X,Y) + b^N\Delta_i^N(X,Y) \leq \Gamma(X,Y)    
\end{equation}
where $\Delta^Q(X,Y) = x_i^Q - y_i^Q$, $Q\in \{N,D\}$, and $\Gamma(X,Y) = \sum_{Q\in \{N,D\}} f_i^Q(Y) - f_i^Q(X)$. 
Concisely we write $X=(x^D_i,x^N_i)$. We can then rewrite (\ref{eq_cond_pricing}) as
\begin{equation}
b^D \leq \frac{\Gamma(X,Y)}{\Delta_i^D(X,Y)} - \frac{b^N\Delta_i^N(X,Y)}{\Delta_i^D(X,Y)}
\end{equation}
for $\Delta_i^D(X,Y) < 0$ and 
\begin{equation}
b^D \geq \frac{\Gamma(X,Y)}{\Delta_i^D(X,Y)} - \frac{b^N\Delta_i^N(X,Y)}{\Delta_i^D(X,Y)}
\end{equation}
for $\Delta_i^D(X,Y) > 0$. 

 Before proving our result on existence of prices, we note that in practice negative prices are possible. That is, users may get paid to charge at certain times. However, there is no logical explanation for both peak and non-peak prices to be negative. In fact, our analysis in Theorem \ref{pricing_thm} below shows that no situation will require peak prices to be negative. We formalize this by defining the feasible price profiles.

\begin{definition}
A price profile $b_j=(b^D_j,b^N_j)$ is feasible if $b^D_j>0$ for all $j\leq n$.
\end{definition}

\begin{theorem}\label{pricing_thm}
Given player-specific price-independent dissatisfaction functions $f_i^Q(\cdot)$, $Q\in \{N,D\}$: 
\begin{enumerate}[I.]
    \item \textit{Linear case}: $f_i^Q(x) = \alpha^Q_i - \beta^Q_ix$, $\alpha^Q_i,\beta^Q_i >0$,\begin{enumerate}
        \item a distributed strategy profile $X$ is a Nash equilibrium if and only if $b=\beta$ when dissatisfaction functions are linear, and, 
        \item there exists feasible prices, including $b=\beta$, for which any non-distributed strategy profile is an equilibrium.
    \end{enumerate} 
    \item \textit{Logistic case}: $f_i^Q(x) = \alpha^Q_i\left(\frac{1}{1+e^{\beta^Q_j\left(2x-1\right)}}\right)$, $\alpha^Q_i,\beta^Q_i>0$ \begin{enumerate}
        \item  there do not exist feasible prices such that a distributed strategy profile is a Nash equilibrium when dissatisfaction functions are logistic, and,
        \item there exists feasible prices for which any non-distributed strategy profile is a Nash equilibrium.
    \end{enumerate}

\end{enumerate}
\end{theorem}
\begin{proof}
For ease of exposition, we drop the player index. 

\textit{Part I, (a) and (b).}\quad In the linear dissatisfaction case where $f^Q = \alpha^Q + \beta^Q x^Q$, $Q\in \{N,D\}$, we know at least one solution exists. Namely, $b^D = \beta^D$ and $b^N = \beta^N$. 

Now, let profiles $X, Y$ and $Z$ so that $\Delta^D(X,Y) < 0 < \Delta^D(X,Z)$ and $\frac{\Delta^N(X,Y)}{\Delta^D(X,Y)} = \frac{\Delta^N(X,Z)}{\Delta^D(X,Z)} = k$, for some $k\in \R$. We get that the only values of $b^N$ and $b^D$ satisfying $X$ to be an equilibrium are those for which
\[
\frac{\Gamma(X,Z)}{\Delta^D(X,Z)} - b^Nk\leq b^D \leq \frac{\Gamma(X,Y)}{\Delta^D(X,Y)} - b^Nk.
\]
Since $f^Q = \alpha^Q + \beta^Q x^Q$, $Q\in \{N,D\}$, we obtain the line
\[
b^D = (\beta^D + \beta^Nk) - b^Nk.
\]
Take any other pair  $Y'$ and $Z'$ so that $\Delta^D(X,Y') < 0 < \Delta^D(X,Z')$ and $\frac{\Delta^N(X,Y)}{\Delta^D(X,Y)} = \frac{\Delta^N(X,Z)}{\Delta^D(X,Z)} = k'$, for some $k'\in \R$ as above we get
\[
b^D = (\beta^D + \beta^Nk') - b^Nk'.
\]
All of these lines $y = (\beta^D + \beta^Nk) - xk$ have a unique point in common. Namely, $x = \beta^N$ and $y = \beta^D$. To see that $b=\beta$ is the only possible solution for a given strategy profile $X$ consider the following four possible representative deviations:
\begin{enumerate}
    \item peak to non-peak deviation from $X$ to $Y$: $\Delta^D(X,Y)=a$, $\Delta^N(X,Y)=-b$, for any $a,b >0$;
    \item non-peak to peak deviation from $X$ to $Y$: $\Delta^D(X,Y)=-a$, $\Delta^N(X,Y)=b$, for any $a,b >0$;
    \item within peak and within non-peak deviation from $X$ to $Y$: $\Delta^D(X,Y)=a$, $\Delta^N(X,Y)=b$, for $a,b >0$. This corresponds to user shifting to more congested periods hence getting less overall quantity.
     \item within peak and within non-peak deviation from $X$ to $Y$: $\Delta^D(X,Y)=-a$, $\Delta^N(X,Y)=-b$, for $a,b >0$. This corresponds to user shifting to less congested periods hence getting more overall quantity.
\end{enumerate}
 
For $X$ to be an equilibrium all the above deviations should not be profitable deviations. That is, (\ref{eq_cond_pricing}) should be satisfied for each of these four deviations. The only feasible price profile that satisfies these four constraints is $b = \beta$. Note that this applies for both distributed and non-distributed cases, and only in non-distributed case there are other prices (not just $b=\beta$) because out of the four representative deviations only two apply to the distributed case. \\

\textit{Part II(a).}\quad For a given strategy profile $X$ consider the following two deviations:
\begin{itemize}
 \item peak to non-peak deviation from $X$ to $Y$: $\Delta^D(X,Y)=a_y$, $\Delta^N(X,Y)=-b_y$, for any $a,b >0$;
    \item non-peak to peak deviation from $X$ to $Y$: $\Delta^D(X,Y)=-a_z$, $\Delta^N(X,Y)=b_z$, for any $a,b >0$;
\end{itemize}

By (\ref{eq_cond_pricing}), these two deviations induce two half-spaces, say $h_1$ and $h_2$. Since $h_1$ is not contained in $h_2$, and vice versa, a feasible price profile only exists if $h_1 \cap h_2$ is non-empty. However, with $a_z<a_y$, $b_z<b_y$, $\frac{a_z}{b_z} = \frac{a_y}{b_y}$ and $\alpha^Q>0$, $h_1$ and $h_2$ are created by two non-intersecting (parallel) half-lines with same slope.

As an example, consider $X=(0.5,0.5)$ and deviations $Y$ and $Z$ such that $a_y=-b_y=0.3$ and $a_z=-b_z=0.4$. The visualization of the half-spaces corresponding to the two deviations are shown in Figure \ref{fig:half_space}, where \textit{red} and \textit{blue} correspond to the two half-spaces.
\begin{figure}
    \centering
    \includegraphics[scale=0.99]{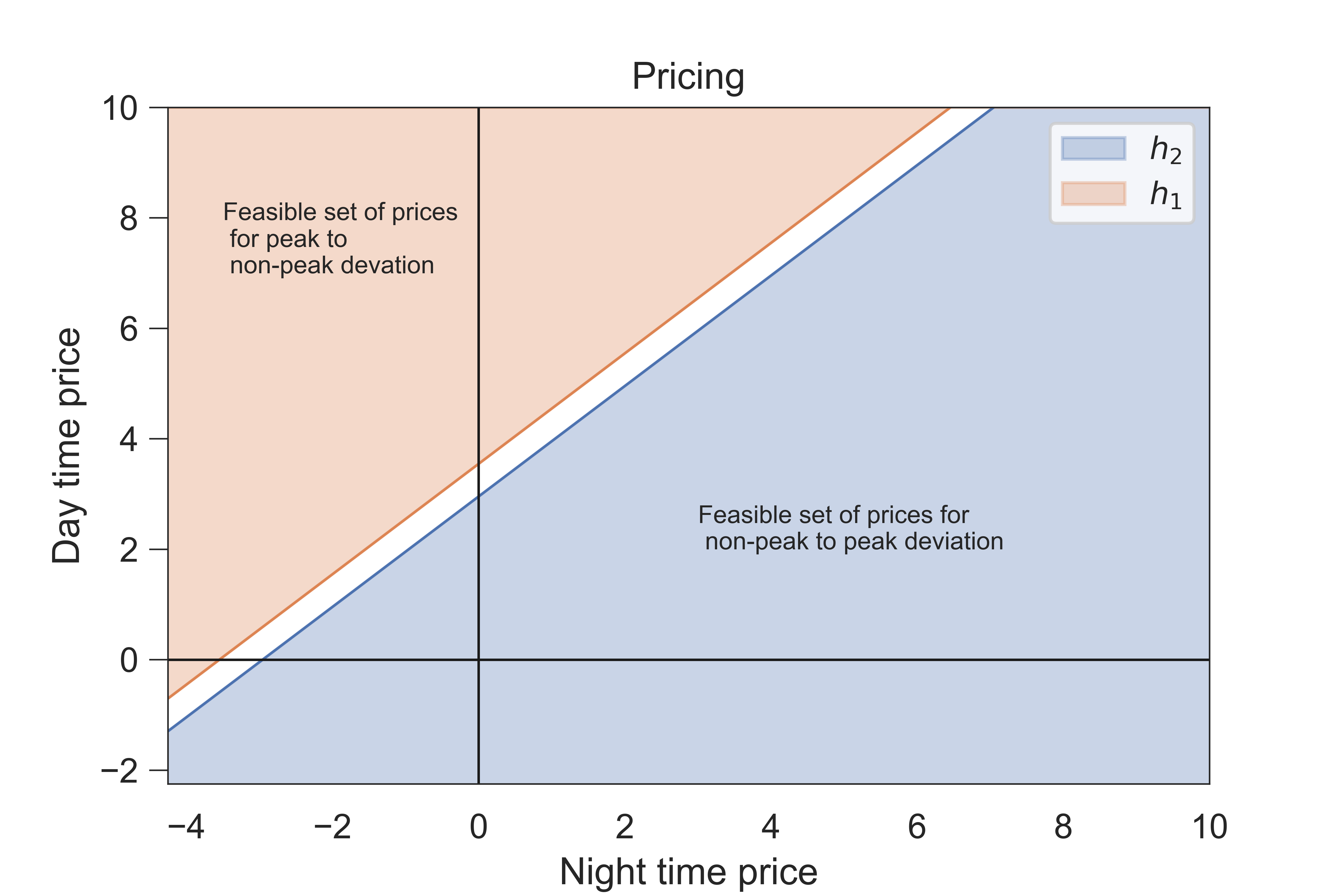}
    \caption{No feasible price profile exists.}
    \label{fig:half_space}
\end{figure}

\textit{Part II (b).} \quad In the case of non-distributed equilibria, there are only two possible type of deviations. For example, without loss of generality, consider a user only charging at peak, then two possible types of deviation are within peak and peak to non-peak. Note that the second type of deviation consists of exactly one deviation unlike the first one. That is, a deviation from peak to non-peak involves moving all charging from peak to non-peak, this is because part movement implies a contradiction of result in part(a). Therefore, for a given $X$ there is only one inequality from (\ref{eq_cond_pricing}) for peak to non-peak deviation (say $Y$), with $\Delta^N(Y)=\Delta^D(Y)$.

Now note that 
\begin{equation*}
    \frac{\Gamma^{Q}}{\Delta^{Q}}\geq 0, \quad Q\in \{N,D\}
\end{equation*}

This implies within peak deviation (say $Z$) will result in one of the following two half-spaces
\begin{equation*}
     b^D\leq \frac{\Gamma^{D}(Z)}{\Delta^{D}(Z)} \quad \text{and} \quad  b^D\geq \frac{\Gamma^{D}(Z)}{\Delta^{D}(Z)}. 
\end{equation*}

The intersection of these half-spaces with half-space defined by peak to non-peak inequality defines the feasible set of prices. Note that this intersection may also lead to negative non-peak prices. 

\end{proof}

\begin{remark}[Implication of Theorem \ref{pricing_thm}, Linear dissatisfaction rate] Implication of this result is that to achieve a desirable (distributed) equilibrium, price discrimination, in other words, personalized pricing is necessary. This also shows that price-based charging control is only possible due to heterogeneity in dissatisfaction levels. However, without the knowledge of dissatisfaction parameters this is impossible. Therefore, our analysis suggests adaptive price experimentation can reveal user dissatisfaction and hence will enable converging to a player specific price profile. Learning (or query) complexity of linear dissatisfaction is an interesting question for further study, such questions have been investigated in routing games context in \textcite{bhaskar2014achieving} for linear and polynomial latencies. 
\end{remark}

\begin{remark}[Implication of Theorem \ref{pricing_thm}, Logistic dissatisfaction rate]
Theorem \ref{pricing_thm} implies that when dissatisfaction are logistic, the only type of equilibria are of non-distributed type, that is, each player either chooses peak or non-peak times. This is what we observe in example 2. Note that herding is non-distributed equilibrium if it is one. One feasible pricing scheme that induces a given outcome as equilibrium can be obtained by setting
\begin{equation*}
    b= \left(\frac{\alpha_j^D{\Gamma}^D_J(Z)}{{\Delta}^D_j(Z)},  -\frac{\left(\left(\alpha^N_j\Gamma^N_j+\alpha^D_j\Gamma^D_j\right)\left(\Delta^D_j(Y)\right)+\left(\Delta^D_j(Z)\right)\Gamma^D_j\alpha^D_j\right)}{\left(\Delta^D_j(Y))\right)\left(\Delta^D_j(Z)\right)}\right)
\end{equation*}
\end{remark}

\section{Mediation/Co-ordination}\label{cce_sec}

To answer our final question we consider studying CCEs which inherently capture the role of an aggregator. To the best of our knowledge CCEs have not been explored, except in \textcite{chakraborty2014demand}, where authors frame decentralized charging as demand response game and show a bound on CCE's Price of Anarchy as a corollary of their main result. However, theoretically analyzing CCEs in EV game is difficult due to discrete strategy sets and the structure of the cost function. For this reason we consider the continuous version of the EV game. To differentiate from the discrete version we will call this continuous version the \textit{C-EV game}. Furthermore, we consider a symmetric setting where all users experience same dissatisfaction and hence are charged same prices akin to setting considered in \textcite{lin2021minimizing}. Formally, the game is as follows, Player $i$ selects a quantity $q_i\in [0,M]$ in peak time (peak time plug-in) and $M-q_i$  in non-peak time, where $M$ is the individual all-day demand for all players. We define the peak and non-peak cost functions with linear dissatisfaction as before as follows:

\begin{equation}
    c^D_i(q_i,q_{-i}) = b^Dq_i\left(1-\frac{Q_D}{M_D}\right) + a^Dq_i\left(\frac{Q_D}{M_D}\right) + r_D, \text{ and }
\end{equation}
\begin{equation}
    c^N_i(q_i,q_{-i}) = b^N(M-q_i)\left(1-\frac{W-Q_D}{M_N}\right) + a^N(M - q_i)\left(\frac{W-Q_D}{M_N}\right) + r_N.
\end{equation}
\noindent
where $Q_D = \sum_{j=1}^N q_j$ is the total quantity requested in peak time by all players, $M_D$ and $M_N$ are normalized constants associated with available capacities in peak and non-peak times, and $W (=nM)$ is the total demand of all players. The first term in $c^D_i(q_i,q_{-i})$ is the direct cost (which is increasing in $q_i$ and decreasing in $Q_D$) and the second term is the peak time cost of dissatisfaction (which is increasing in $Q_D$). Note that in this case we are expressing cost as a function of $\frac{q_iQ_D}{M_D}$ which is shortfall or expected power not received, thus, making our analysis much simpler. Given a strategy profile $(q_i, q_{-i})$ we have the overall cost of Player $i$ is 
\begin{equation}
     c_i(q_i,q_{-i}) = c^D_i(q_i,q_{-i}) +  c^N_i(q_i,q_{-i}).
\end{equation}
 Note that $c_i^D$ and $c_i^N$ are both convex when $a^D > b^D$ and $a^N > b^N$, hence $c_i$ is convex too. The cost function can be re-expressed as:


\begin{equation}
    c_i(q_i,q_{-i}) = Aq_i + Bq^2_i + Bq_i\sum_{j\neq i}q_j - R\sum_jq_j + \eta,
\end{equation}\noindent
where $A = b^D-\frac{a^NW}{M_N} - \frac{b^N(M_N-W)}{M_N}$, $B =\left( \frac{a^D-b^D}{M_D} - \frac{a^N-b^N}{M_N}\right)$, $R=\frac{(a^N-b^N)M}{M_N}$ and $\eta = r_D+r_N$.  As proved in the previous section and owing to linear dissatisfaction, selecting prices accordingly will always ensure the existence of Nash equilibria. Hereafter, we assume for the choice of prices and dissatisfaction parameters, Nash equilibrium exists. In which case, Nash equilibrium in C-EV game can be characterized as:
\begin{equation}\label{nash-qty}
    q^{Nash}_i = \max \left\{ 0,\min \left\{\frac{R-A}{B(n+1)},M\right\}\right\}.
\end{equation}


We will now analyze CCEs in this game. 
It should be noted that in the context of CCE, the only decision 
EV users  have to make (simultaneously and independently) is whether to
`commit' to the mediating agency (e.g., aggregator) or not. Once this decision is made, the `committed' users act according to the recommendation of the agency. The concept of CCE requires stronger commitment of the EV users to a
controlled charging scheme, in the sense that the EV users have to decide (simultaneously)
whether to abide by the recommendations of the agency or act on their own. Formally, a recommendation device is a lottery (probability) distribution over all possible outcomes. In other words, a mediating agency recommends charging strategies to EV users based on this distribution.

 Let $Q_i=[0,M]$ and $\mathbb{Q}=\prod_{i}Q_{i}$, with
generic elements q$_{i}$ and q respectively, and continuous cost functions $c_{i}:$ $\mathbb{Q}\rightarrow \mathbb{R}$, $i=1,\ldots ,n$. Let $\mathcal{L}(\mathbb{Q})$, with generic element $L$, and $\mathcal{L}%
(Q_{i})$, with generic element $\ell _{i}$, denote the sets of probability
measures on $\mathbb{Q}
$ and $Q_{i}$, respectively. For simplicity, let the expectation of $c_{i}$ with respect to $L$ be
denoted by $c_{i}(L)$.

With this notation the following definition establishes when a distribution is CCE. 

\begin{definition}\label{cce_def}
\label{cce-def} \textit{A coarse correlated equilibrium (CCE) of the game }$%
G $\textit{\ is a distribution }$L\in \mathcal{L}(
\mathbb{Q}
)$\textit{\ such that }$c_{i}(L)\leq c_{i}(q_{i},L^{-i})$ for all $q\in 
\mathbb{Q}
$.
\end{definition}

If $L$ is the
distribution of a symmetric random variable $Z=(Z_{1},\ldots ,Z_{n})$,
consider respectively the expected values of $Z_{i}$, $Z_{i}^{2}$, and $%
Z_{i}\cdot Z_{j\neq i}$, $i=1,\ldots ,n$, and denote them as below;: 
\begin{align*}
\gamma_1 & =E_{L}[Z_{i}]\text{,} \\
\gamma_2 & =E_{L}[Z_{i}^{2}]\text{ and} \\
\gamma_3 & =E_{L}[Z_{i}\cdot Z_{j}].
\end{align*}%

\begin{theorem}\label{main_cce_result}
CCE in C-EV game coincide with Nash equilibria.
\end{theorem}

To prove Theorem \ref{main_cce_result}, we first show that the CCE equilibrium constraint (as in Definition \ref{cce_def}) for the $n$-player quadratic game can be completely expressed in terms of these three moments of a symmetric probability distribution. 
\begin{lemma}
\label{n-quad-cce-cons} Any symmetric probability distribution $L\in \mathcal{L}(\mathbb{%
Q})$ is a CCE of the C-EV game if and only if 
\begin{equation*}
\min_{z\geq 0}\left\{ (A-R+B(n-1)\gamma_1 )z+Bz^{2}\right\} \geq
(A-R)\gamma_1 +B(\gamma_2 + (n-1)\gamma_3)  \text{;}
\end{equation*}
\end{lemma}
\begin{proof}

When all EV users commit, Player $i$'s cost function is
\begin{equation}
    c_i(L) = (A-R)\gamma_1 + B(\gamma_2 + (n-1)\gamma_3) - R(n-1)\gamma_1 + \eta,
\end{equation}
on the other hand when Player $i$ chooses not to commit while all others commit to $L$, then Player $i$'s cost is:
\begin{equation}
    c_i(z,L^{-i}) = Az + Bz^2 + Bz(n-1)\gamma_1 - Rz - R(n-1)\gamma_1 + \eta
\end{equation}
For equilibrium, we need $c_i(L) - c_i(z,L^{-i})\leq 0$, which gives the result. 
\end{proof}

The set of CCE can now be charaterised as given in the following proposition:
\begin{lemma}\label{cce-set}
The set of CCE in C-EV game are characterized by the following set of constraints:
\begin{align}
 \gamma_2  \geq \gamma_3; & \quad  \gamma_2 \leq M\gamma_1  \label{wcce2} \\
\gamma_2 +(n-1)\gamma_3 & \geq n\gamma_1 ^{2} \label{wcce3} \\
\gamma_2 +(n-1)\gamma_3 & \leq \frac{1}{B}\left( (R-A)\gamma_1 - \frac{1}{4B}\left((R-A) - B(n-1)\gamma_1 \right)^2 \right) \label{wcce4}\\
\gamma_1 ,\gamma_3 & \geq 0;  \label{wcce6}
\end{align}
\end{lemma}
\begin{proof}
From \textcite{dokka2022equilibrium}, we have (\ref{wcce2})--(\ref{wcce3}). For completeness we reproduce some of the arguments below.
For $L$ to be feasible, it should be
true that the variance-covariance matrix underlying $L$ is positive
semi-definite (PSD). Omitting the subscript $L$ for ease of notation, let $%
Y_{i}=Z_{i}-\gamma_1 $ for all $i$: 
\begin{align*}
Var(Z_{i})=E[Y_{i}^{2}]& =\gamma_2 -\gamma_1 ^{2}=\gamma_2 ^{\ast }\text{,} \\
Cov(Z_{i},Z_{j})=E[Y_{i}Y_{j}]& =\gamma_3 -\gamma_1 ^{2}=\gamma_3 ^{\ast }\text{.}
\end{align*}%
We need to express a necessarily PSD matrix with $\gamma_2 ^{\ast }$ on the diagonal and $%
\gamma_3 ^{\ast }$ on the off-diagonal. This means that we have for all 
$x\in \mathbb{R}^{n}$ 
\begin{equation}
\gamma_2 ^{\ast }\left(\sum_{1}^{n}x_{i}^{2}\right)+2\gamma_3 ^{\ast }\left(\sum_{1\leq i\leq
j\leq n}x_{i}x_{j}\right)\geq 0\text{.}  \label{h_psd1}
\end{equation}%
this holds if and only if 
\begin{equation}
\gamma_2 ^{\ast }\geq \gamma_3 ^{\ast }\text{ and }\gamma_2 ^{\ast }+(n-1)\gamma_3
^{\ast }\geq 0\text{,}
\end{equation}%
where $\gamma_2 ^{\ast }\geq 0$ but $\gamma_3 ^{\ast }$ can be positive or
negative. Note that $\gamma_2 ^{\ast }\geq 0$ is necessary, and if $\gamma_2 ^{\ast }=0$
then we need $\gamma_3 ^{\ast }=0$ as well. Assume now $\gamma_2 ^{\ast }>0$.
\medskip

\textit{Case 1: $\gamma_3 ^{\ast }\geq 0$.} \newline
In this case, we can write (\ref{h_psd1}) as 
\begin{equation}
\left(1-\frac{\gamma_3 ^{\ast }}{\gamma_2 ^{\ast }}\right)\left(\sum_{1}^{n}x_{i}^{2}\right)+\frac{%
\gamma_3 ^{\ast }}{\gamma_2 ^{\ast }}\left(\sum_{1}^{n}x_{i}\right)^{2}\geq 0\text{,}
\end{equation}%
which holds if and only if $\gamma_2 ^{\ast }\geq \gamma_3 ^{\ast }$.

\textit{Case 2: $\gamma_3 ^{\ast }<0$.}\newline

In this case, (\ref{h_psd1}) is 
\begin{equation}
\left(1-\frac{-\gamma_3 ^{\ast }}{\gamma_2 ^{\ast }}\right)\left(\sum_{1}^{n}x_{i}^{2}\right)\geq \frac{%
-\gamma_3 ^{\ast }}{\gamma_2 ^{\ast }}\left(\sum_{1}^{n}x_{i}\right)^{2}\text{.}
\end{equation}%
If we fix the sum $\sum_{1}^{n}x_{i}$, the minimum of the LHS\ above is
achieved when all $x_{i}$ are equal, so that the inequality holds for all $x$
if and only if it holds for $x$ on the diagonal, i.e., 
\begin{equation*}
1+\frac{-\gamma_3 ^{\ast }}{\gamma_2 ^{\ast }}\geq n\frac{-\gamma_3 ^{\ast }}{\gamma_2
^{\ast }}\Longleftrightarrow \gamma_2 ^{\ast }+(n-1)\gamma_3 ^{\ast }\geq 0\text{.%
}
\end{equation*}
Combining both cases and switching back to $\gamma_2 $ and $\gamma_3 $, we get (\ref{wcce2})--(\ref{wcce3}). 
Combining this with Lemma \ref{n-quad-cce-cons}, we have the desired result.  
\end{proof}

\begin{proof}[Proof of Theorem \ref{main_cce_result}] 
Note that the (\ref{wcce3}) and (\ref{wcce4}) define two half-spaces defined by two parallel lines, this is because they both have same left hand sides. Therefore, there are CCE other than Nash only if there is a value of $\alpha \neq \frac{R-A}{B(n+1)}$ such that the intersection of these two half-spaces is non-empty. However, this is not possible as the maximum of the following concave function, obtained by taking the difference of right hand sides of (\ref{wcce3}) and (\ref{wcce4}),
\begin{equation}
    \frac{1}{B}\left( (R-A)\gamma_1 - \frac{1}{4B}\left((R-A) - B(n-1)\gamma_1 \right)^2 \right) - n\gamma_1^2,\label{rhs_diff}
\end{equation}
is equal to $0$ and is achieved at $\gamma_1 =  \frac{R-A}{B(n+1)}$. Note that only when (\ref{rhs_diff}) is positive the half-spaces overlap more than just meeting at the boundary. This implies the only feasible solution to the system (\ref{wcce3})--(\ref{wcce4}) is obtained when CCE decisions ($\gamma_1$,average peak time quantity) coincides with that of Nash.
\end{proof}

\begin{remark}[Implication of Theorem \ref{main_cce_result}]
By Theorem \ref{main_cce_result}, EV users do not gain anything as against selfish behavior leading to Nash equilibria by committing to a mediating agency. This is in stark contrast to many important economic situations where CCEs are better equilibria than Nash, see  \textcite{moulin2014improving}.  The result in Theorem \ref{main_cce_result} implies that for EV users to commit to a mediating agency (eg., aggregator), heterogeneity in their dissatisfaction rates is necessary. Note that by definition C-EV game is symmetric. While our result does not prove if mediated communication in asymmetric case, when each user has its own dissatisfaction rate and is priced accordingly, can lead to CCEs which are different from Nash, it shows that if it does, the underlying reason is the asymmetry itself. It would be interesting to study CCEs in asymmetric version, given its complexity is out of scope of this paper and hence we will pursue it for future study. 
\end{remark}

\begin{remark}
Interestingly, the Nash equilibrium characterized by (\ref{nash-qty}) shows that in when all users are symmetric the equilibrium is of distributed type. This implies herding is the only possible non-distributed equilibrium in this case. Contrast this with example 1, where the equilibrium is of non-distributed type. This shows that in the asymmetric case there are more equilibria and more importantly, possibly, both distributed and non-distributed.
\end{remark}




\printbibliography

\end{document}